%% file: l.tex
\documentclass{amsart}
\usepackage{fullpage,amssymb,epic,eepic,epsfig,amsmath}

\theoremstyle{plain}

\newtheorem{theorem}{Theorem}
\newtheorem{proposition}{Proposition}[section]

\newtheorem{observation}[proposition]{Observation}
\newtheorem{corollary}[proposition]{Corollary}

\theoremstyle{definition}
\newtheorem{definition}[proposition]{Definition}

\theoremstyle{remark}

\def\printname#1{
	\if\draft y
		\smash{\makebox[0pt]{\hspace{-0.5in}
			\raisebox{8pt}{\tt\tiny #1}}}
	\fi
}

\newlength{\standardunitlength}
\catcode`\@=11
\long\def\@makecaption#1#2{%
     \vskip 10pt

\setbox\@tempboxa\hbox{
       \small\sf{\bfcaptionfont #1. }\ignorespaces #2}%
     \ifdim \wd\@tempboxa >\captionwidth {%
         \rightskip=\@captionmargin\leftskip=\@captionmargin
         \unhbox\@tempboxa\par}%
       \else
         \hbox to\hsize{\hfil\box\@tempboxa\hfil}%
     \fi}
\font\bfcaptionfont=cmssbx10 scaled \magstephalf
\newdimen\@captionmargin\@captionmargin=2\parindent
\newdimen\captionwidth\captionwidth=\hsize
\catcode`\@=12

\def\B{\mathcal B}

\def\K{\mathcal K}

\def\a{\alpha}

\def\b{\beta}

\begin{document}

\title{DNA-Inspired Information Concealing}

%
%
%
%
%

\author{Lukas~Kencl
        and~Martin~Loebl}
\thanks{M. Loebl is with the Department of Applied Mathematics and Institute of Theoretical Informatics (ITI), Charles University, Prague, Czech Republic.}
\thanks{L. Kencl is with the Research and Development Centre (RDC), Czech Technical University, Prague, Czech
Republic}



\begin{abstract} 

Protection of the sensitive content is crucial for extensive information 
sharing. We present a technique of information concealing, 
based on introduction and maintenance of families of repeats. 
Repeats in DNA constitute a basic obstacle for its reconstruction by hybridisation. 
Information concealing in DNA by repeats is considered in \cite{letter}.
\end{abstract} 

\maketitle

\section{Introduction} 
\label{S_Introduction}

Contemporary computer systems may be distributed and
may consist of many interconnected processing units or a large
number of networked computer subsystems. In addition contemporary
digital networks may consist of a large number of end- and
intermediate- nodes. In all these systems, information, in the
form of the sequences over some alphabet of symbols, is
circulating or being stored. The entity controlling a subsystem or
a node is often unwilling or prohibited to share this
information-sequences with other nodes. However, sharing of some
reduced local information might be very useful for purposes of
security, stability and various analysis of the system
performance, and for data mining. Such analysis might for example
allow to identify frequently appearing segments by performing
approximate statistical analysis on segment frequency, allowing to
detect replicating malicious code-worms. It also allows to
identify segments-markers of computer viral infection, by
detecting patterns existing in some database of malicious
sequences. Such databases are used e.g. in contemporary intrusion
detection systems or spam filters. It has been shown that being
able to perform pattern matching against only fixed-length
prefixes or substrings of longer sequences can provide approximate
hints as to the presence of suspicious
content~\cite{ramaswamy06fingerprinting}. Likewise, established
worm detection techniques such as Autograph~\cite{Kim:ATA04} or
EarlyBird~\cite{singh04earlybird} are based on counting frequency
of small blocks of a fixed size.

Sharing of reduced local information among the members of an
interconnected computer system or communication network thus helps
to discover attacks earlier. Affected parts may be isolated and
further attack spread prevented. The benefits of sharing local
information may be reaped in case of existence of a computational
information processing, which preserves local information (e.g.
all segments of certain maximal length) and makes impossible to
reconstruct longer or sensitive parts of the information
sequences.

We call such information processing {\em concealing}. The systems
which conceal information and share the concealed information are
likely to possess a competative advantage in the form of
robustness, attack resistance and immunity due to ability to
exchange, publish and protect information. Clearly, any
information concealing algorithm needs to address two conflicting
goals:
\begin{enumerate}
\item preserving \emph{presence} and, possibly, \emph{frequency
rank} of segments of given size (making spam identification and
worm detection still possible), while

\item making reconstruction of content longer than the predefined
limit computationally hard (e.g. disabling interpretation or
understanding of the private content).
\end{enumerate}

\subsection{Main contribution}
\label{sub.contrr}

The main contribution of this paper is
\begin{itemize}
\item Formulation of the information concealing problem \item
Presentation of an information concealing algorithm \item Analysis
of the algorithm and a proof of the hardness of reconstruction of
the input sequence
\end{itemize}

\section{Related Work} 
\label{sec.related}

\subsection{Repeats in DNA}
\label{sub.euler} Our inspiration comes from an important feature
of eukaryotic DNA, namely that it contains various {\em repeat
families}, and that their presence constitutes a basic difficulty
in DNA reconstruction by hybridisation \cite{P}.

A large proportion of eukaryotic genomes is composed of DNA
segments that are repeated either precisely or in variant form
more than once. Highly repeated segments are arranged in two ways:
as tandem arrays or dispersed among many unlinked genomic
locations. As yet, no function has been associated with many of
the repeats \cite{MB}. In the paper \cite{letter} which
accompanies this paper, the authors propose that in eukaryotes the
cells have DNA as a depositary of concealed genetic information
and the genome achieves the self-concealing by accumulation and
maintenance of repeats. The protected information may be shared
and this is useful for the development of intercellular
communication and in the development of multicellular organisms.

The assertion that the repeats are maintained in DNA in a
programmed way for self-concealing explains basic puzzling
features of repeats: the uniformity along with the polymorphism of
the repeated sequences; the freedom of the repeated DNA to adopt
quite different primary sequences in closely related species;
apparent non-functionality of the precise amount or the precise
sequence of the repeats.

The containment of repeats versus DNA sequencing problem is
receiving extensive attention of biologists, computer scientists
and mathematicians (see \cite{G}, \cite{P}, \cite{Bo}).

\subsection{Repeats versus DNA reconstruction}
\label{sub.euler} We explain the basic idea of concealing by
repeats in this subsection. Assume we are given a collection $\K$
of segments of DNA. Each segment $S$ from $\K$ is divided into two
parts, the initial part $S(I)$ and the terminal part $S(T)$. We
thus may write $S= S(I)|S(T)$. This is an artificial assumption
imposed only for the clarity of the presentation.

A {\em reconstruction} of $\K$ is a sequence of its segments so
that the terminal part of each segment agrees with the initial
part of next segment in the sequence. If several of these initial
and terminal parts coincide, there may be an exponential number of
possible reconstructions.

Let us consider a very simple example. Let $\K$ be the following
collection of segments, where the initial and the terminal parts
are divided by the vertical line:
$$
A|B, B|A, A|C, C|A, B|C, C|B.
$$
the following sequences are some of the possible reconstructions:
$$
ABACBCA, ACABCBA, BACABCB, ABCACBA.
$$
In this simple example, even unlimited computational power is
useless to anybody who wants to obtain the correct reconstruction
from the many possible reconstructions. This phenomenon may well
be described in terms of the de Bruin graph: this graph has a node
for each segment which is an initial or a terminal part of an
element of $\K$. For each segment $S$ of $\K$ there is an arrow (a
directed edge) from $S(I)$ to $S(T)$.

\begin{figure}[h]
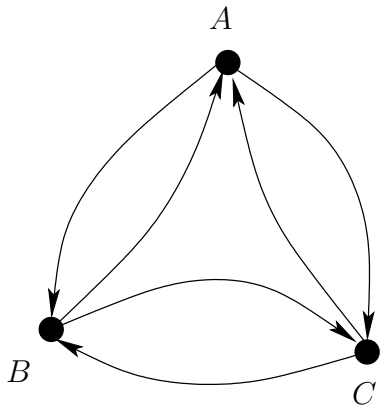

\begin{center}
\input bg.pstex_t
\end{center}
\caption{De Bruin graph for $\K$}
 \label{fig.bg}
\end{figure}

The possible reconstructions now correspond to the walks on the de
Bruin graph so that each directed edge is traversed exactly once.
These walks are usually called Euler walks. If a node of the de
Bruin graph has more than one outgoing incident directed edge,
then locally there are several independent ways to traverse these
edges. The number of the Euler walks of the de Bruin graph is
therefore typically exponential in the number of these nodes (see
\cite{Bo} for the calculations).

\subsection{Concealing in Information and Communication Technologies}
\label{sub.ict}

The concept of hiding private or sensitive data but preserving
some form of structural information has been studied recently in
various sub-domains of ICT. Some techniques concentrate on hiding
the originator of information, i.e. \emph{anonynimization}, other
focus on enabling particular functions over the data that can be
shared among multiple partners, such as \emph{private matching}.

\subsubsection{Concealing Network Data}


An anonymization scheme over the network packet \emph{IP
addresses} called CryptoPan~\cite{xu02cryptopan} preserves the
prefix hierarchy of the original addresses, while making them
computationally hard to reconstruct by using hashing. This in turn
allows to share network traces (with packet headers only), with
preservation of the prefix hierarchy.

Similarly, in~\cite{rexford04routeranonymization}, structure of
the router configuration files and data is preserved, while the
actual values are obfuscated.

A technique to process and transform the network \emph{packet
payload} has been proposed in~\cite{paxson03anonymization}. This
method uses dictionaries of important sequences that are valuable
from the data mining perspective and should be preserved, while
encrypting the rest of the information with cryptographically
strong hash function. This technique performs well in terms of
data protection, however, it only allows to study content portions
pre-determined by a known list, and thus does not allow to study
the payload to detect previously unknown content, such as e.g.
malicious subsequence.

The popular Bloom filter~\cite{bloom70hashcoding} approach is used
in constructing the Hierarchical Bloom Filter \emph{payload
attribution} technique~\cite{shanmugasundaram04bloom}. A Bloom
filter can store (incompletely but efficiently) input items (which
can be substrings) and easily answer set membership queries. It
consists of $k$ hash functions, each associating one of $m$
numbers to each input item. Set membership queries exhibit no
false negatives, but can have false positives.


Payload attribution with a Hierarchical Bloom Filter stores
segments of network packet payloads with their IP source and
destination addresses. Each payload is cut into segments
$s_1,\dots, s_n$. The $s_i$'s are stored in a Bloom filter of
level $0$, {\it the pairs} $s_1s_2, s_3s_4,\dots$ in Bloom filter
of level $1$, quadruples in level $2$, and so on. A query on an
excerpt of payload, which may consist of several consecutive
blocks, may answer the source and destination address by running
through consecutive hierarchy levels.

The authors propose deployment at network concentration points.
Privacy protection is to be achieved by restricting access of
entities that can pose queries, otherwise exhaustive attacks might
lead to payload reconstruction.

\subsubsection{Private matching} \label{sub.matching}

Private matching~\cite{li04privatematching,
agrawal03privatesharing} focuses on the problem of two entities
trying to find common data elements in their databases, without
revealing private information. The basic property (and difference
to the general information concealing problem) is that only two
parties are involved; a multiparty solution is a future work
suggestion. Further problems are asymmetry in the sequence of
information exchange among the parties and needed presumption of
honesty ('semihonesty' in the paper).

%

Private matching is a special case of cryptography theory of {\it
multi-party computation}: $m$ parties want to compute function $f$
on their $m$ inputs. In the {\it ideal model}, where a trusted
party exists, the parties give their inputs to the trusted
authority, it calculates $f$, and returns the result to each
party. The ideal model assumes an ideal situation: for example, no
protocol can prevent a party to change its input before the
communication is started. A secure multiparty computation protocol
{\it emulates} what happens in the ideal model.

Paper~\cite{li04privatematching} also introduces 'data ownership
certificates' to modify the private matching protocols to be
unspoofable. This technique is shown to be useful in a more
practical setting to enable privacy-protecting sharing of e-mail
white-lists in~\cite{kaminsky06reliableemail}.

\subsubsection{Data Masking} \label{sub.masking}

Various techniques of masking, sanitizing and obfuscating data
have been studied to enable test- or third-party development over
sensitive databases (such as the Human Resources data).  After
sanitization, the database remains usable - the look-and-feel and
some relations and distributions are preserved - but the
information content is secure. The used techniques include
masking, shuffling, substitution, number-variance, encryption etc
~\cite{datamasker}. These techniques share a similar goal with
information concealing, but focus on structured data without the
need of preserving the local information.

\subsubsection{Data Mining and Anonymization} \label{sub.mining}

In data mining, anonymization mechanisms (obfuscating the
originator or the private part of the data) are currently studied
intensively. Privacy mechanisms can be classified into several
categories, according to where they are deployed during the life
cycle of the data. The mechanism proposed in this paper falls into
the category where the individuals trust no one but themselves,
and they conceal  their respective data before they make them
available for sharing. The existing algorithms in this
category~\cite{MS, AS, AH, EGS, MS1} are called local
perturbation; they are based on different ideas then the
concealing procedure proposed in this letter.

In another category, data publishing, data are anonymized at a
central server; the individuals are required to trust this
server~\cite{RS}. Anonymization in social networks is studied
in~\cite{HMJW}.

An important theoretical foundation for data anonymity and
originator protection was laid in~\cite{sweeney02k-anonymity}. The
k-anonymity model for protecting privacy allows holders to release
their private data without being distinguishable from at least k-1
other individuals also in the release.

\subsubsection{Steganography} \label{sub.stega}

This form of information hiding ~\cite{johnson01steganography,
cox07steganography} is a related art and science of writing hidden
messages in such a way that no one apart from the sender and
intended recipient realizes there is a hidden message; this
nowadays includes concealment of digital information within
computer files. Comparably, steganalysis is the art of detecting
the hidden information.

Steganography is a mature science, in particular focussing on the
domain of Digital Rights Management (DRM), where various
'watermarking' or 'tamper-proofing' techniques may seamlessly
embed extra information about the origin of a digital work within
itself. This is a different goal than the proposed information
concealing. While the embedded information may be well concealed
and thus very hard to reconstruct, data mining of such information
would not be generally possible. However, it would be very
interesting to apply the steganalysis tools to the information
concealing 'attacker problem' (see Section~\ref{sec.concprob}) and
it certainly belongs among our future work.

\subsubsection{Information Retrieval} \label{sub.retrieval}

The "attacker problem" of concealed string reconstruction (see
Section~\ref{sec.concprob}) has a strong connection to the problem
of information retrieval~\cite{manning08retrieval}, where
probabilistic information about the expected string (e.g. natural
text) may be used to derive further information or assist text
reconstruction.

\subsection{Segments shuffling}
\label{sub.shuffling}
Finally we mention that our first attempt to solve the anonymization problem~\cite{zamora06}
was using random permutations of a collection of short overlapping segments. This method
however by itself doesnot lead to concealing the original data information. It is shown in this paper
that in order to sufficiently extend the families of repeats of the resulting sequence and make the concealing successful, other procedures need to be performed as well.
In particular the overlapping segments containing complete
local information need to be prolonged by attaching additional short segments to their beginning and/or their end. The shuffling permutation also needs to satisfy some properties. This is described in the rest of the paper.



\section{Information Concealing Problem}
\label{sec.concprob}

We introduce formally the information-sequence concealing problem.
Let $|\omega|$ denote the number of symbols (length) of sequence
$\omega$. The \emph{sequence concealing} problem is the following:
Given a sequence $\omega$ and a small positive integer $k$, we
want to transform $\omega$ to another sequence $\omega_F$ so that:
\begin{enumerate}
\item[I.-] If $s$ is a segment of $\omega$ with $|s| \leq k$, then
$s$ is a segment of $\omega_F$.

\item[II.-] It is computationally hard to reconstruct sequence
$\omega$ from $\omega_F$.

\item[III.-] The length of $\omega_F$ is linear in $|\omega|$.

\item[IV.-] It is also desirable that with low probability, a
segment not in $\omega$ appears in $\omega_F$, and that relative
frequency (i.e., frequency rank)
 of segments of $\omega$ of a given length is preserved in $\omega_F$.
The precise statement of these two conditions is however strongly
application dependent.

\end{enumerate}

Given the statement of the information concealing problem, the key
issue is how much information about $\omega$ can an attacker
deduce from $\omega_F$; let us call this issue the {\em attacker
problem}.

Clearly, the answer to the attacker problem is
application-dependent. If the input sequence $\omega$ is very
restrictive, e.g. if a short prefix uniquely determines larger
part of $\omega$ and the k-segments of $\omega$ may be
distinguished within the larger k-segment superset of $\omega_F$,
then inevitably large part of input $\omega$ may be reconstructed
from $\omega_F$. In quite a number of practical situations (DNA
sequence, computer program, sound and video trace,
 text on non-specific topic), however, this is not the case. Moreover, for restrictive
input sequences, we can perform preparatory procedures (as
procedure $S$ described below) which make the input sequence less
specific.

This partly justifies the following {\em consistency assumption}
concerning the attacker problem which we need to make in order to
carry the security analysis of the concealing algorithm.

\begin{proposition}
\label{p.cons} The complete input of the attacker problem, i.e.
all the useful information an attacker has about $\omega$, is
$\omega_F$, the length $|\omega|$, the length $k$ of the preserved
segments and the concealing algorithm used in obtaining
$\omega_F$.
\end{proposition}

Thus, an attacker may use the list of frequencies of repeats of
segments of $\omega_F$ along with the knowledge of the concealing
algorithm to attempt the reconstruction of $\omega$.

\section{Concealing by repeats}
\label{sec.conrr} The input of the problem is a sequence over an
alphabet. We first turn it into a cyclic sequence by connecting
its beginning and end.

Next we describe five procedures which are used in the algorithm.
 The basic pattern of all
the procedures is the same and may be described as follows: the
input is a cyclic sequence $\omega$. First, $\omega$ is
partitioned into consecutive disjoint {\em blocks}. Then the
terminal part of the preceding block of length $o$ (the {\em
overlap}) is added in front of each block. The resulting segments contain all the studied local
information; depending on the procedure, these segments will also contain some excess information 
which is vital in a proposed composition of the procedures which forms our concealing algorithm.
Next, a segment called {\em dust} can but neednot be added behind each segment. The enhanced blocks are
called the {\em cards}. The last step consists in arranging the
cards into the output cyclic sequence $\omega_F$.

The first procedure $S$ has a preparatory character in the
concealing algorithm. Several runs of $S$ have the role of
breaking the local sequential order in the input sequence.

\subsection{Procedure $S(\omega,o,lb,ub)$}
\label{sub.s}

Its input is a cyclic sequence $\omega$, and it has parameters
$o,lb, ub$; $o$ stands for the size of the overlap, $lb$ is for
lower bound of the length of a block, and $ub$ is for the upper
bound of the length of a block. The procedure $S(\omega,o,lb,ub)$
is defined by 1.-4. below.

\begin{enumerate}
\item[1.] We partition (sometimes we say that we {\em cut})
$\omega$ into consecutive disjoint {\it blocks} $P_1,\dots, P_m$
such that the length of each $P_i$ is chosen at random between
$lb, ub$.
\item[2.] We add overlap of length $o$ in front of each block. The
overlapping segments thus contain all the original sub-segments
 of length up to $o+1$.
\item[3.] The blocks enhanced by the overlaps now start and end
 with the corresponding overlaps. If these were arranged into a cyclic sequence,
the overlaps would neighbor. This may help an attacker in
reconstruction. To break the neighborhood relationship of the
overlaps, we may add dust (a randomly chosen segment) behind each
block. Adding dust is optional and application dependent. A
natural restriction is that the dust is a segment of the input
sequence and that the average length of dust is $1/2(lb+ub)- o$ to
match the average length of the segments complementing the
overlaps. However, depending on applications, and the stringency
of condition $[IV]$ of the sequence concealing problem, length of
dust may be different and the dust need not be a segment of the
input sequence.
\item[4.] We arrange the resulting cards randomly into a cyclic
sequence.
\end{enumerate}


As an illustration we perform $S$ on an example input sequence:

\

\begin{center}
 \fbox{\parbox{0.96\columnwidth}{{\em Example 1: Procedure $S(\omega,o,lb,ub)$}

\

Input $\omega =$ '{\bf the aim of this paper is to present an
information concealing algorithm}', parameters $o= 3$, $lb= 4$,
$ub= 6$.
\begin{enumerate}
\item[1.] First, the input sequence is partitioned randomly into
blocks of length $4,5$ or $6$. The blocks are divided by '$+$'
below:
\newline
'{\bf the ai$+$m of t$+$his $+$pape$+$r is $+$to pr$+$esent$+$ an
i$+$nfor$+$matio$+$n co$+$ncea$+$ling $+$algor$+$ ithm$+$}'
\item[2.] Next we add overlap (of length $o= k-1= 3$) in front of
each block:
\newline
'{\bf thmthe ai$+$ aim of t$+$f this $+$is pape$+$aper is $+$is to
pr$+$ present$+$ent an i$+$n infor$+$formati$+$tion co$+$
concea$+$cealing $+$ingalgor$+$gorithm$+$}'
\item[3.] Next we add the dust behind each block (of length
approximately 2), and we get the cards:
\newline
'{\bf thmthe aip$+$ aim of tim$+$f this con$+$is pape in$+$aper is
a$+$is to pro p$+$ presentese$+$ent an ilgo$+$n infori
$+$formatifo$+$tion co $+$ concea ci$+$cealing pa$+$ingalgor
p$+$gorithmap$+$}'
\item[4.] Finally the output is given by arranging the cards in a
random order (here we use the order
$14,9,10,13,5,3,12,1,6,4,7,11,8,15,2$):
\newline
'{\bf ingalgor pn infori formatifocealing paaper is af this
conconcea cithmthe aipis to pro pis pape in presentesetion co ent
an ilgogorithmap aim of tim}'
\end{enumerate}

}}
\end{center}

\subsection{Procedure $S^1(\omega,lb,ub)$}
\label{sub.S1}

Procedure $S^1(\omega,lb,ub)$ is as $S$ but the overlap is always
the whole preceding block - typically exceeding the size needed to
preserve the studied local information (this excess is used in the
composition of the procedures forming our concealing algorithm). Hence, if
the blocks are
$$
\omega= P_1P_2P_3\dots P_m,
$$
then the cards of $S^1$ are $P_1P_2, P_2P_3,\dots, P_mP_1$.

Each $P_i$ appears once as initial segment and once as terminal
segment of each card. Hence, the cyclic consecutive order of the
cards of $S^1$
 may be described by a permutation $\pi$ of $1,\dots, m$; for further discussions it turns out useful
to define such permutation so that it assigns, to each terminal
block of a card, the initial block of the next card. By {\em
permutation of $1,\dots, m$} we mean a bijection from set
$\{1,\ldots, m\}$ onto itself. If $\pi$ is a permutation then
$\pi^{-1}$ denotes the inverse permutation ($\pi(x)= y$ if and
only if $\pi^{-1}(y)= x$). Hence, in our formalism, card
$P_{i-1}P_i$ is followed by card $P_{\pi(i)}P_{\pi(i)+1}$.

The output of $S^1$ thus always has form
$$
P_1P_2P_{\pi(2)}P_{\pi(2)+ 1} \dots
P_{\pi^{-1}(1)-1}P_{\pi^{-1}(1)}.
$$
For instance, if we have $m=3$ then the cards are $P_1P_2, P_2P_3,
P_3P_1$ and a shuffling which results in sequence
$P_1P_2P_3P_1P_2P_3$ is described by permutation $\pi(1)=2,
\pi(2)=3, \pi(3)=1$.

\subsubsection{Acceptable permutations}
\label{sub.ap}

For our purposes, not all permutations $\pi$ are acceptable; let
us formally denote by $\mathcal A$ the set of all the {\it
acceptable permutations}. To define $\mathcal A$, we first
introduce an auxiliary bipartite graph $G(\pi)$.

\begin{definition}
\label{def.auxg} Graph $G(\pi)$ has vertex-set $V= V_1\cup V_2$
where $V_1= \{u_1,\ldots, u_m\}$ and $V_2=\{v_1,\ldots, v_m\}$.
The edge-set of $G(\pi)$ is the union of three disjoint perfect
matchings of the vertex-set, namely:

\begin{enumerate}
\item [1.] The perfect matching $M_1$ consisting of the edges
$\{u_i, v_i\}$.
\item [2.] The perfect matching $M_2$ consisting of the edges
$\{u_{i+1}, v_i\}$.
\item [3.] The perfect matching $M_3$ consisting of the edges
$\{u_{\pi(i)}, v_i\}$.
\end{enumerate}
\end{definition}

\begin{definition}
\label{def.augg} We construct a directed graph $G'(\pi)$ from
$G(\pi)$ by first directing each edge of $M_2\cup M_3$ from $V_2$
to $V_1$, and then contracting each edge of $M_1$.
\end{definition}

\begin{definition}(of set ${\mathcal A}$ of all acceptable permutations )
\label{def.acc} Permutation $\pi$ is acceptable ($\pi \in
{\mathcal A}$) if and only if the following two conditions are
satisfied:
\begin{enumerate}
\item [1.] The directed graph $G'(\pi)$ has a directed eulerian
closed walk where the edges of $M_2$ and $M_3$ alternate. This
condition is equivalent to saying that permutation $\pi$ describes
a rearrangement of the cards of $S^1$ into a sequence. \item [2.]
In the auxiliary graph $G(\pi)$, the union of the perfect
matchings $M_2\cup M_3$ contains many (at least $m/c$ where $c\geq
2$ is a small constant) cycles. This condition is added in order
to make the reconstruction of the input sequence hard; see the
sections below.
\end{enumerate}
\end{definition}

The following observation about the graph $G(\pi)$ will be used in
the analysis of the attacker problem.

\begin{observation}
\label{o.grat} Let $G(\pi)$ be as in Definition \ref{def.auxg}.
For $v_i \in V_2$ let $s(v)= P_i$ be its associated segment. Then
we have the following equality between cyclic sequences:
$$
P_1P_2P_3\ldots P_m= s(M_2(1))s(M_2(2))\ldots s(M_2(m)),
$$
where $M_2(i)$ denotes the vertex of $V_2$ connected with $u_i \in
V_1$ by an edge of $M_2$.
\end{observation}


For illustration we perform $S^1$ on the output sequence of the
previous example (which would be the natural use of $S^1$, as
described later):

\

\begin{center}
 \fbox{\parbox{0.96\columnwidth}{{\em Example 2: Procedure
$S^1(\omega,lb,ub)$}

\

Input $\omega =$ '{\bf ingalgor pn infori formatifocealing paaper
is af this conconcea cithmthe aipis to pro pis pape in
presentesetion co ent an ilgogorithmap aim of tim}', parameters
$lb= 6$ and $ub= 8$.
\newline
First, the input sequence is partitioned randomly into blocks of
length $6,7$ or $8$. The blocks are divided by '$+$' below:
\newline
'{\bf ingalgo$+$r pn inf$+$ori for$+$matifo$+$cealing$+$ paape$+$r
is a$+$f this c$+$onconce$+$a cithm$+$the aipi$+$s to pro$+$ pis
p$+$ape in$+$ present$+$esetion $+$co ent $+$an
ilg$+$ogorith$+$map ai$+$m of tim$+$}'
\newline
Next we add overlap in front of each block. For procedure $S^1$
the overlap is always the whole preceding block. We get the
following cards; to make the example easier to understand we
indicate by '*' the division of each card into two blocks:
\newline
'{\bf m of tim*ingalgo$+$ingalgo*r pn inf$+$r pn inf*ori for$+$ori
for*matifo$+$matifo*cealing$+$cealing* paape$+$ paape*r is a$+$r
is a*f this c$+$f this c*onconce$+$onconce*a cithm$+$a cithm*the
aipi$+$the aipi*s to pro$+$s to pro* pis p$+$ pis p*ape in$+$ape
in* present$+$ present*esetion $+$esetion *co ent $+$co ent *an
ilg$+$an ilg*ogorith$+$ogorith*map ai$+$map ai*m of tim$+$}'
\newline
Finally the output is given by rearranging the cards by an
acceptable permutation, i.e. by a permutation whose corresponding
bipartite graph consists of a lot of cycles. The smallest length
of a cycle is $4$. It is not difficult to see that the following
permutation $\pi$ creates nine $4-$cycles and one $6-cycle$. In
the following description of $\pi$, the cycles are grouped
together; for instance the first $4-$cycle has edges $(v_1,
u_{10}), (v_9, u_2), (v_1, u_2), (v_9, u_{10})$. The first two of
them belong to perfect matching $M_3$, the last two belong to
perfect matching $M_2$.
\newline
$[\pi(1)= 10, \pi(9)= 2]; [\pi(2)= 6, \pi(5)= 3]; [\pi(3)= 9,
\pi(8)= 4]; [\pi(7)= 13, \pi(12)= 8]; [\pi(14)= 11, \pi(10)= 15];
[\pi(11)= 18, \pi(17)= 12]; [\pi(19)= 14, \pi(13)= 20]; [\pi(16)=
21, \pi(20)= 17];[\pi(15)= 19, \pi(18)= 16];[\pi(21)= 7, \pi(4)=
5, \pi(6)= 1]$.
\newline
Hence the final sequence (for ease of understanding we preserve
the separation symbols '*', which in reality would not be
present):.
\newline
'{\bf ingalgo*r pn inf paape*r is a pis p*ape inthe aipi*s to prof
this c*onconcer pn inf*ori foronconce*a cithm present*esetion  m
of tim*ingalgoa cithm*the aipian ilg*ogorithape in*
presentogorith*map aico ent *an ilgesetion *co ent s to pro* pis
pmap ai*m of timr is a*f this cmatifo*cealingori
for*matifocealing* paape}'

}}
\end{center}

\subsection{Procedure $S^{1+}(\omega,lb,ub)$}
\label{S1+}

If the input of the procedure $S^1$ comes from several runs of the
preparatory procedure $S$ described above, then we need to modify
$S^1$ in order to make its output generic, that is to
intentionally preserve the attacker-confusing overlaps. This
modified procedure is called $S^{1+}$.

We recall that $S^1$ repeats the whole blocks $P_i$, i.e. the
output of $S^1$ is the cyclic sequence
$$
P_1P_2P_{\pi(2)}P_{\pi(2)+ 1} \dots
P_{\pi^{-1}(1)-1}P_{\pi^{-1}(1)}.
$$

We assume that the input $\omega$ of $S^{1+}$ comes from repeated
runs of procedure $S$ and so $\omega$ contains a lot of segments
of length $o$ (the overlaps of runs of $S$) repeated at least
twice; let us denote by $R$ the set of all these segments.

Procedure $S^{1+}$ starts as $S^1$ by partitioning of $\omega$
into blocks
$$
P_1,P_2,\dots, P_m.
$$
The blocks of $S^{1+}$ cut some of the segments from $R$. To
reflect this, we write $P_i=r^T_{i-1}Q_ir^I_i$ where
\begin{itemize}
\item Segment $r^T_{i-1}$ is an empty segment or a terminal
segment of an element of $R$ cut by the partition between blocks
$P_{i-1}$ and $P_i$. \item Segment $r^I_i$ is an empty segment or
an initial segment of an element of $R$ cut by the partition
between blocks $P_i$ and $P_{i+1}$.
\end{itemize}
Summarizing this notation we write
$$
P_1P_2\dots P_m= Q_1r_1Q_2r_2Q_3r_3\dots r_{m-1}Q_mr_m,
$$
where each $r_i$ is such an element of $R$ that is cut by the
blocks of $S^{1+}$, or an empty segment. Each
$P_i=r^T_{i-1}Q_ir^I_i$ where $r_i=r^I_ir^T_i$.

The first difference of $S^1$ and $S^{1+}$ is that the overlaps of
$S^{1+}$ are not the whole preceding blocks. Instead, the overlap
added in front of block $P_{i+1}$ is $Q_ir^I_i$. Hence, block
$P_{i+1}$ with the overlap added in front of it has form
$Q_ir_iQ_{i+1}r^I_{i+1}$.

To make the cards of $S^{1+}$ more generic (see the same step in
the description of Procedure $S$), we change each such
$Q_ir_iQ_{i+1}r^I_{i+1}$ into $Q_ir_iQ_{i+1}r'_{i+1}$ where
$r'_{i+1}$ is obtained from $r^I_{i+1}$ by adding a segment so
that $r'_{i+1}$ has length $o$ and is repeated elsewhere in
$\omega$.

Summarising, the output of $S^{1+}$ has form
$$
Q_1*Q_2*Q_{\pi(2)}*Q_{\pi(2)+ 1}*\ldots
*Q_{\pi^{-1}(1)-1}*Q_{\pi^{-1}(1)}*,
$$
where each $*$ stands for a segment of length $o$ which is
repeated (at least) twice in this output, or the empty string.
More specifically, if $*$ follows segment $Q_i$ then it is equal
to $r_i$ or to $r'_i$.


\subsection{Procedure $S^2(\omega,o)$}
\label{sub.S2}

Let $S^2(\omega,o)$ be as follows: we assume its input is an
output of $S^1$, i.e. it is the cyclic sequence
$$
P_1P_2P_{\pi(2)}P_{\pi(2)+ 1} \dots
P_{\pi^{-1}(1)-1}P_{\pi^{-1}(1)}.
$$
Note that in this sequence, each block $P_i$ appears twice.
Procedure $S^2$ first {\em cuts} each $P_i$ randomly into $P_i^1,
P_i^2$ so that length of $P_i^1$ is at least $o$, i.e. the whole
overlap of length $o$, which we denote by $o_i$, is contained in
$P_i^1$. The trick of the concealing algorithm is that {\em both
copies of each $P_i$ are cut in the same way!} Let $o_iP_i^2$
denote $P_i^2$ with the added overlap.

For example, if $P_i$ is equal to 'abcdefghijkl' and $o=3$ then a
possible cut of $S^2$ is 'abcde$+$fghijkl'; $P_i^1$ is equal to
'abcde', $P_i^2$ is equal to 'fghijkl' and $o_iP_i^2$ is equal to
'cdefghijkl'.

We may describe the set of the cards of $S^2$ as the {\em
disjoint} union of two sets $C_1\cup C_2$, where
$$
C_1= \{o_1P_1^2P_2^1, o_2P_2^2P_3^1, \dots, o_mP_m^2P_1^1\}
$$
and
$$
C_2= \{o_1P_1^2P_{\pi(1)}^1, o_2P_2^2P_{\pi(2)}^1, \dots,
o_mP_m^2P_{\pi(m)}^1\}.
$$

We remark here that the cards of $C_1$ correspond to the edges of
perfect matching $M_2$ of graph $G(\pi)$ and the cards of $C_2$
correspond to the edges of perfect matching $M_3$ of $G(\pi)$ (see
Definition \ref{def.auxg}).

Finally $S^2$ arranges $C_1\cup C_2$ into a random cyclic
sequence.


For illustration we perform $S^2$ on the output sequence of the
previous example 2 (which would be the natural use of $S^2$, as
described later):

\

\begin{center}
 \fbox{\parbox{0.96\columnwidth}{{\em Example 3: Procedure
$S^2(\omega, o)$}

\

Input $\omega =$ '{\bf ingalgo*r pn inf paape*r is a pis p*ape
inthe aipi*s to prof this c*onconcer pn inf*ori foronconce*a cithm
present*esetion  m of tim*ingalgoa cithm*the aipian ilg*ogorithape
in* presentogorith*map aico ent *an ilgesetion *co ent s to pro*
pis pmap ai*m of timr is a*f this cmatifo*cealingori
for*matifocealing* paape}', parameter $o=3$.

A consistent partitioning into blocks is indicated below:
\newline
'{\bf inga$+$lgo*r pn i$+$nf pa$+$ape*r is$+$ a pis$+$ p*ape$+$
inthe a$+$ipi*s to $+$prof this $+$c*onco$+$ncer pn i$+$nf*ori$+$
foronco$+$nce*a ci$+$thm pre$+$sent*eseti$+$on  m of
$+$tim*inga$+$lgoa ci$+$thm*the a$+$ipian i$+$lg*ogori$+$thape$+$
in* pre$+$sentogori$+$th*map $+$aico e$+$nt *an i$+$lgeseti$+$on
*co e$+$nt s to $+$pro* pis$+$ pmap $+$ai*m of $+$timr is$+$ a*f
this $+$cmati$+$fo*ceal$+$ingori$+$ for*mati$+$foceal$+$ing*
pa$+$ape}'
\newline
Next we add overlap (of length $o$) in front of each block (and we
delete the 'helpful symbol' *):
\newline
'{\bf apeinga$+$ngalgor pn i$+$n inf pa$+$ paaper is$+$ is a
pis$+$pis pape$+$ape inthe a$+$e aipis to $+$to prof this $+$is
conco$+$nconcer pn i$+$n infori$+$ori foronco$+$nconcea ci$+$
cithm pre$+$presenteseti$+$etion  m of $+$of timinga$+$ngalgoa
ci$+$ cithmthe a$+$e aipian i$+$n ilgogori$+$orithape$+$ape in
pre$+$presentogori$+$orithmap $+$ap aico e$+$o ent an i$+$n
ilgeseti$+$etion co e$+$o ent  to $+$to pro pis$+$pis pmap $+$ap
aim of $+$of timr is$+$ is af this $+$is
cmati$+$atifoceal$+$ealingori$+$ori formati$+$atifoceal$+$ealing
pa$+$ paape}'
\newline
Finally we rearrange the cards in a random order. The resulting
sequence is as follows:
\newline
'{\bf n inforio ent s to ori formati paapen ilgesetipis papen
ilgogoringalgor pn iapeingaof timr is is af this presentesetin inf
paealingoriealing papresentogorietion  m of atifocealap aim of
ngalgoa cie aipisan iof timingaatifocealis cmatipis pmap
orithapeis concoori foroncoto pro pise aipis to  paaper isnconcer
pn ietion co e is a pis cithm preo ent an ito prof this nconcea
ciap aico eape inthe aorithmap  cithmthe aape in pre}'

}}
\end{center}

\subsection{Procedure $S^{2+}(\omega,o)$}
\label{sub.S2+}

We assume its input is an output of $S^{1+}$. This procedure is
defined analogously as $S^2$ with the only difference that the
{\em cuts} are performed to segments $Q_i$ instead of segments
$P_i$.

\section{The concealing algorithm}
\label{s.ca}

Let the input string be $\omega$, and the length of the preserved
segments be $k$. We consider two scenarios, {\em weak concealing}
and {\em strong concealing}, depending on the nature of the input.
We perform the {\em weak concealing algorithm} if the input is
nonspecific, i.e., short segments have many possible alternative
prolongations, or there does not exist any outside knowledge about
the likelihood of presence of some segments in the input (e.g. an
English text).

The {\em weak concealing algorithm} may be described as
$$
\omega_F= S^2(S^1(\omega,3k/2,2k),k-1).
$$
We choose to have the block length in $S^1$
 longer and to overlap the whole blocks in $S^1$ since we want to ensure that the {\em cuts} of $S^2$
may be done in the same way in each of the two copies of the
blocks $P_i$.

The {\em strong concealing algorithm} may be written as

$$
\omega_F= S^{2+}(S^{1+}(S \ldots S(\omega,k-1,k,3k/2))),3k/2,2k),k-1),
$$
where the number of repetitions of procedure $S$ is application
specific.

\section{Analysis of the concealing algorithm} 
\label{S_Properties}

\begin{observation}
\label{o.kpres} The concealing algorithm preserves all segments of
length $k$ present in the input sequence $\omega$ within the
output sequence $\omega_F$.
\end{observation}
This observation is straightforward as whenever any of the above
procedures cuts the input string, an overlap of length at least $k-1$ is
added in front of the segment following the cut, thus preserving
all subsegments of length $k$ which would otherwise be partitioned
by the cut.

It is also straightforward that both weak and strong concealing
algorithms are linear in $|\omega|$ if we have
\begin{itemize}
\item Access to a generator of random permutations of the numbers
less than $|\omega|$, \item Access to a generator of random
elements of ${\mathcal A}$ (see Definition \ref{def.acc}).
\end{itemize}

A random permutation may be generated in linear time (see
\cite{KN}). We will not discuss the complexity of generating
random elements of ${\mathcal A}$. Instead, we specify a large
subset $\B$ of ${\mathcal A}$ such that generating a random
element of $\B$ may be reduced to generating a random permutation
of a number less than $|\omega|$.

Each element of $\B$ may be constructed as follows: we take any
permutation $\pi$ of $m/2$ (we assume $m$ is even) and we consider
the pairing $P(\pi)$ of $\{1,2,\ldots m\}$ given by
$(1,\pi(1)+m/2),\ldots (m/2,\pi(m/2)+ m/2)$. This pairing may be
looked at as an involution $i(\pi)$ (a permutation $\a$ is
involution if $\a(\a(x))= x$ for each $x$) on $m$. Finally, we get
element $\b= \b(\pi)$ of $\B$ by shifting $i(\pi)$ by $1$, i.e., by
letting $\b(a)= i(\pi)(a)+1$ modulo $m$; we have an additional condition that
$O^j(1)\neq 1$ for $j<m$ and $O(a)= \b(a)+1$. This condition makes sure that 
the first condition of the definition of the acceptable permutation is satisfied.
The following observation
is straightforward.

\begin{observation}
\label{o.expp}
$$
|\B|\leq (m/2-1)!
$$
Further, the graphs defined by a permutation from $\B$ are
disjoint unions of $m/2$ cycles of length $4$. Generating a random
element from $\B$ is as hard as choosing a random permutation of
$m/2$.
\end{observation}


The following observation is also straightforward.

\begin{observation}
\label{o.len} The length of the output of each of the  procedures 
applied to input $\omega$ is linear in $|\omega|$. For example, for $S$ and $S^1$ it is 
$2|\omega|$.
\end{observation}

\section{Hardness of the attacker problem}
\label{s.hard} We recall that the attacker problem introduced in
Section \ref{sec.concprob} (see also Proposition \ref{p.cons})
reads:

{\em How much information about $\omega$ can an attacker deduce
from $\omega_F$, $|\omega|$, $k$ and the knowledge of the
concealing algorithm?}

For instance, the attacker can try to get all the overlaps of
$S^2$ since {\em assuming $\omega_F$ has no accidental repeats}
these overlaps appear exactly four times in $\omega_F$ and no
other segment is like that. The attacker may partition $\omega_F$
into cards as indicated by all these overlaps. She gets a
collection of cards, with $(k-1)$-length segments marked in the
beginning and the end of each card. The attacker wants to overlap
these marked segments. Depending on whether {\it $\omega_F$ has
accidental repeats}, the attacker possibly {\em cuts} in more
places than were the original cards used in the algorithm. Hence,
in her collection of cards some overlaps should not have been
considered, and some segments have overlaps with more than one
other card. These considerations naturally specify the {\bf
domino} and {\bf donkey} problems.

In more realistic situation the attacker does not know the correct
list of cards of $S^2$ and hence she needs to choose which
$4-$repeats to ignore. We may assume that she has some hints as to
which overlaps are 'likely' ok. This is the situation we model by
the following problem.

{\bf Shortest domino row problem (SDRP).} Assume we are given a
collection of dominoes (domino will mean a rectangle partitioned
vertically into two squares, where one is initial and the other
one is terminal), and we are also given a graph on the squares.
This graph should be interpreted as the graph of hints. We want to
put all the dominoes into a row, so that if two consecutive
squares are connected by an edge of the graph, we can put one
square on top of the other (i.e., identify them). The aim is to
make the resulting row as short as possible, i.e. to satisfy as
many hints as possible.

Let us define the (de Bruijn-type) graph $G= (V,E)$ where $V$ is
the set of all the squares, and $E$ is the set of the dominoes:
edge $e_i$ connects the squares of domino $Q_i$. The following
observation is straightforward.

\begin{observation}
\label{o.eut} There is a natural bijection between the set of the
Euler circuits (eulerian closed walks) of $G$ and the set of all
the circular sequences consistent with the overlapping dominoes
$Q_1,\dots, Q_m$.
\end{observation}

\begin{theorem}
\label{thm.dom} The SDRP is search-NP-complete.
\end{theorem}
\begin{proof}
Assume that in the auxiliary graph, there is edge between two
squares if they are equal, but not all such edges are there. This
is exactly consistent with our interpretation. Now, in the
reformulation with the de Bruijn graph and the Euler circuit, this
corresponds to the problem that we are given a graph, with some
transitions between neighboring edges recommended, and we want to
find an Euler circuit with as many recommended transitions as
possible. A particular instance is that some transitions are
forbidden, and we want to find out whether Euler circuit where all
the transitions are allowed exists. This is known to be NP
complete (\cite{GJ}).

We have in fact a {\it search instance} of this problem: we know
that such an Euler circuit exists, and we want to find it. There
is a standard trick which shows that the decision problem is
polynomial if the search problem is polynomial:

Assume there is a polynomial algorithm $A$ that solves the search
version, and let its running time be $n^{10}$, say. To solve the
decision problem, we apply $A$ to an input. It either finds the
right Euler circuit and then the answer is YES, or it runs longer
than $n^{10}$, and then the answer is NO.
\end{proof}

In the {\bf donkey problem} we assume that {\it $\omega_F$ has no
accidental repeats}. What the attacker gets? There are two
versions of the algorithm. Let us first consider the {\it strong
concealing} where the preliminary step is performed.

\begin{enumerate}
\item[1.] As described above, using the $4-$repeats of length
$k-1$ of $\omega_F$, the attacker gets the cards of $S^{2+}$, i.e.
$C_1\cup C_2$, where
$$
C_1= \{o_1Q_1^2r_1Q_2^1, o_2Q_2^2r_2Q_3^1, \dots,
o_mQ_m^2r_mQ_1^1\}
$$
and
$$
C_2= \{o_1Q_1^2r'_1Q_{\pi(1)}^1, o_2Q_2^2r'_2Q_{\pi(2)}^1, \dots,
o_mQ_m^2r'_mQ_{\pi(m)}^1\}.
$$
\item[2.] The attacker also gets each $Q_i^1$ and each $o_iQ_i^2$
since
   these are exactly maximal initial and terminal segments of the cards above
which are repeated twice in $\omega_F$. \item[3.]
 By matching the overlaps, the attacker gets each pair $Q_i^1Q_i^2$ since
the overlap $o_i$
   in $o_iQ_i^2$ is a terminal segment of $Q_i^1$ and we may assume that
   these cannot be misinterpreted.
\item[4.] What the attacker gets from the initial applications of
procedure $S$? Each of their overlaps (of length $k-1$) appears at
least twice in the input of $S^{1+}$. Moreover most of the {\em
cuts} of the procedures $S$ are different. Let us recall here that
among these overlaps may be also the dust. Procedures $S^{1+}$ and
$S^2$ {\em cut} into some of these. Those cut will remain
2-repeats, those not cut may gain repeats. Moreover, $S^{1+}$
introduces dust in the border of each card: this adds 2-repeats of
strings of length $k-1$ undistinguishable from the 2-repeats
coming from initial procedures $S$.

\end{enumerate}

In case weak concealing algorithm is applied, the attacker has
$1.,2.,3.$ where $Q_i^2r_i$ and $Q_i^2r'_i$ are replaced by
$P^2_i$ and $Q^1_i$ is replaced by $P^1_i$.

The next proposition summarises the possible types of repeats
introduced by the algorithm.

\begin{proposition}
\label{p.corr} All the repeats of $\omega_F$ generated by the weak
or strong concealing algorithm are those described in
$1.,2.,3.,4.$.
\end{proposition}

\begin{corollary}
\label{p.conc} All the useful information for the attacker problem
is $|\omega|$, $k$, and $1.,2.,3.,4.$.
\end{corollary}

The information $1.,2.,3.$ may be described by the auxiliary
bipartite graph $G(\pi)$ defined in Definition \ref{def.auxg}.

If the weak concealing algorithm is applied, information $[4.]$
does not exist. The attacker problem is thus reduced to the
following:

{\bf The donkey-decision problem.} The input is a bipartite graph
$G$ where the vertices in both parts $V_1, V_2$ are ordered. Let
$V_1=\{u_1,\dots, u_m\}$ and $V_2=\{v_1,\ldots, v_m\}$. Moreover a
segment $s(v)$ of length at least $3k/2$ is associated with each
element of $V_2$. The set of the edges of $G$ is formed by  a
disjoint union of two perfect matchings $M_2, M_3$. The attacker
needs to reconstruct string
$$s(M_2(1))s(M_2(2))\dots s(M_2(m)),$$
where $M_2(i)$ is the vertex of $V_2$ connected with $u_i\in V_1$
by an edge of $M_2$.

\medskip

The difficulty of the donkey-decision problem is the following:
bipartite graph $G$ is a union of two edge-disjoint perfect
matchings. Each vertex of $G$ thus has degree $2$ and $G$ is a
union of disjoint cycles. To solve the donkey-decision problem,
one needs to {\em choose the correct perfect matching}
independently in each of these cycles (namely, the perfect
matching induced by $M_2$). This is impossible, and the list of
all the possibilities is almost always exponential in the number
of the cycles, since each of the cycles has two perfect matchings.
This is analysed precisely below, when we speak about the {\em
feasible solutions}.

Next we argue that, when the strong concealing algorithm is
applied, the attacker problem is reduced to the donkey-decision
problem too. The attacker is left with the statistics of the
repeats of $\omega_F$. Here comes the reason why we introduced the
dust in $S^{1+}$: it is to make sure that the 2-repeats appear
symmetric for both matchings $M_2,M_3$. This hides the repeats
introduced by the initial applications of procedure $S$. The
information of $[4.]$ is thus useless. We obtain:

\begin{proposition}
\label{prop.att} The attacker problem for both strong and weak
concealing is reduced to the analysis of the donkey-decision
problem.
\end{proposition}

A {\em feasible solution} to the donkey-decision problem is any
sequence $s(M(1))s(M(2))\dots s(M(m))$, where $M$ is any perfect
matching of the input bipartite graph $G$. In order to solve the
donkey-decision problem, one needs to choose, from the pool of
these feasible solutions, the unique correct one. Next we argue
that unless the input to our problem is extremely restrictive,
there is an exponential number of the competative solutions.

The bipartite graphs $G$ coming from ${\mathcal A}$ have at least
$2^{m/c}$ perfect matchings. The output sequences of two perfect
matchings $M,N$ may still be equal: if the cycle has length $4$,
this happens if and only if the two vertices $v_i, v_j$ of $V_2$
in each $4-$cycle in which $M,N$ differ have the same associated
segment ($s(v_i)= s(v_j)$ as defined in Observation \ref{o.grat}).

For instance, if all the vertices of $V_2$ have the same
associated segment, then there is only one competative solution.
This extreme situation may happen if the input $\omega$ is a
sequence of repetitions of one symbol only.

If two symbols may appear in the segments (of length at least
$3k/2$) associated with the vertices of $V_2$, then the
probability that in $a$ $4$-cycles the corresponding pairs of
strings are indistinguishable is $2^{-3ka/2}$. Hence with only
exponentially small probability there is less than an exponential
number of feasible solutions.

\section{Conclusion}
\label{S_Conclusion}

We define the information concealing problem and propose an
algorithm to solve it. It is based on the intuition coming from
the difficulties of DNA reconstruction by hybridisation. The
algorithm may be efficiently implemented. In analysing the amount
of information leaked by the concealing algorithm to an attacker
(this is called the {\em attacker problem} in the paper), we first
consider the case that the output contains random repeats; this
leads to the {\it domino problem} which is shown to be
NP-complete. Even if the attacker solves the domino problem, she
is faced with the {\it donkey problem} which is reduced to the
{\it donkey-decision problem}. It is shown that with high
probability the donkey-decision problem has an exponential number
of feasible solutions among which the attacker needs to choose the
correct one.

\end{document}

%% file: bg.pstex_t
\begin{picture}(0,0)%
\includegraphics{bg.pstex}%
\end{picture}%
\setlength{\unitlength}{4144sp}%
\begingroup\makeatletter\ifx\SetFigFont\undefined%
\gdef\SetFigFont#1#2#3#4#5{%
  \reset@font\fontsize{#1}{#2pt}%
  \fontfamily{#3}\fontseries{#4}\fontshape{#5}%
  \selectfont}%
\fi\endgroup%
\begin{picture}(2255,2478)(481,-2110)
\put(1711,209){\makebox(0,0)[lb]{\smash{{\SetFigFont{12}{14.4}{\familydefault}{\mddefault}{\updefault}$A$}}}}
\put(496,-1906){\makebox(0,0)[lb]{\smash{{\SetFigFont{12}{14.4}{\familydefault}{\mddefault}{\updefault}$B$}}}}
\put(2566,-2041){\makebox(0,0)[lb]{\smash{{\SetFigFont{12}{14.4}{\familydefault}{\mddefault}{\updefault}$C$}}}}
\end{picture}%